\documentclass[preprint,12pt,3p]{elsarticle}

\usepackage{amssymb}

\usepackage{amsthm}
\usepackage{graphicx}
\usepackage{amsmath}
\usepackage{hyperref}
\usepackage{graphicx}

\usepackage{pgf,tikz}

\usetikzlibrary{snakes}
\tikzstyle{printersafe}=[snake=snake,segment amplitude=0 pt]

\usetikzlibrary{arrows}
\usetikzlibrary{shapes}
\usepackage{multirow}
\usepackage{latexsym}
\usepackage{setspace}
\PassOptionsToPackage{boxed,section}{algorithm}
\usepackage{algorithm}
\usepackage{algorithmic}
\usepackage{enumerate}
\usepackage{amsmath}			
\usepackage{amssymb}			
\usepackage{url}
\usepackage{setspace}
\usepackage{tikz}
\usetikzlibrary{arrows}
\usetikzlibrary{shapes}
\usetikzlibrary{decorations.markings}
\usepackage{geometry}
\usepackage{bbm}

\newtheorem{theorem}{\em Theorem}

\newtheorem{definition}{\em Definition}

\newtheorem{corollary}{\em Corollary}

\journal{Sample Journal}

\begin{document}

\begin{frontmatter}

\title{Parameterized algorithms for Partial vertex covers in bipartite graphs}

\author[label1]{Vahan Mkrtchyan\corref{cor1}}
\address[label1]{Gran Sasso Science Institute, L'Aquila, Italy}

\cortext[cor1]{Corresponding author}

\ead{vahan.mkrtchyan@gssi.it}

\author[label5]{Garik Petrosyan}
\address[label5]{Department of Informatics and Applied Mathematics, Yerevan State University, Yerevan, Armenia}
\ead{garik.petrosyan.1@gmail.com}

\author[label6]{K. Subramani}
\address[label6]{West Virginia University, Morgantown, USA}
\ead{k.subramani@mail.wvu.edu}


\begin{abstract}
 In this paper, we discuss parameterized algorithms for variants of
 the partiall vertex cover problem. Recall that in the classical vertex
 cover problem (VC), we are given a graph $G = (V, E)$
 and a number $K$ and asked if we can cover the edges $e \in E$, using
 at most $K$ vertices from $V$. In the Partial vertex cover problem (PVC),
 in addition to the parameter $K$, we are given a second parameter $K'$
 and the question is whether we can cover at least $K'$ edges $e \in E$
 using at most $K$ vertices from $V$.  The weighted generalizations of the
 VC and PVC problems are called the Weighted vertex cover (WVC) and the
 Partial weighted vertex cover problem (WPVC)  respectively. In the WPCV problem,
 we are given two parameters $R$ and $L$, associated respectively with the
 vertex set $V$ and edge set $E$ of the graph ${\bf G}$. Additionally, we are
 given non-negative integral weight functions for the vertices and the edges.
 The goal then is to cover edges of total weight at least $L$, using vertices
 of total weight at most $R$. (In the WVC problem, the goal is  to cover
 all the edges with vertices whose total cost is at most $R$).
  Observe that the variants of VC mentioned
 here, viz., PVC, WVC and WPVC are all generalizations of VC and hence
 their  {\bf NP-completeness} follows immediately from the {\bf NP-completeness} of VC. One attack
 on {\bf NP-complete} problems is to devise algorithms that are polynomial, if
 certain selected selected parameters are bounded. Such algorithms, if they
 exist are called parameterized algorithms and if they run in time polynomial
 in the size of the input (but exponential time in the size of the parameter), the
 problem is said to be fixed-parameter tractable. This paper studies several
 variants of the PVC problem and establishes new results from the perspective
 of fixed parameter tractability and {\bf W[1]-hardness.} We also introduce a new problem
 called the Partial vertex cover with matching constraint and show that it is fixed-parameter
 tractable for a certain class of graphs.
\end{abstract}


\end{frontmatter}



\section{Introduction}
   
  In this paper, we study several variants of the vertex cover (VC) problem, from
  the perspectives of parameterized algorithm design and parametric complexity.
  In particular, we consider the partial vertex cover problem, wherein the goal
  is to cover a certain threshold of edges (as opposed to all the edges) using
  the fewest number of vertices. We also look into weighted variants of this

 The rest of this paper is organized as follows: The problems studied in this paper are
 formally described in Section \ref{sop}. In Section \ref{motwork}, we discuss the motivation
 for our work and mention related approaches in the literature. Our main results
 are described in Section \ref{mainres}. A variant of the partial vertex cover problem
 with applications to computational social choice is detailed in Section \ref{matching}.
 We conclude in Section \ref{conc}, by summarizing our results and outlining avenues
 for future research.

\section{Statement of Problems}
\label{sop}

We focus on finite, undirected graphs that have no
loops or multiple edges. As usual, the degree of a vertex is the
number of edges of the graph that are incident to it. The maximum
degree of the graph $G$ is just the maximum of all degrees of vertices
of $G$. A graph $G=(V,E)$ is bipartite, if its vertex set can be
partitioned into two sets $V_1$ and $V_2$, so that each edge of $G$
joins a vertex from $V_1$ to one from $V_2$.

Given a graph $G= ( V, E)$, and a set ${ S \subset V}$ of vertices,
an edge $(i,j) \in { E}$ is \textit{covered} by ${ S}$ if $i \in { S}$
or $j \in { S}$. Let  $E(S)$ to be the set of edges of $G$ that are
covered with at least one vertex of $S$. The classical Vertex Cover
problem (VC) is defined as finding the smallest set $S$ of vertices of
the input graph $G$, so that $E(S)=E$. The vertex cover problem is a
well-known {\bf NP-complete} problem \cite{Kar72}.

 In this paper, we study the following variants of VC:
\begin{enumerate}
\item The Partial vertex cover problem (PVC) - In this problem -
\begin{definition}
Given an undirected graph ${ G=( V,E)}$,  and a vertex-cardinality parameter $K_1$ and an
edge-cardinality parameter $K_2$, is there a subset $V'$ of $V$, such that $|V'| \le K_1$ and
the number of edges of covered by $V'$ is at least $K_2$?
\end{definition}

\item The Weighted Partial Vertex Cover problem (WPVC) - 

\begin{definition}
\label{def:WPVCB} 
Given an undirected graph ${ G=( V,E)}$,  weight-functions $c: { V} \rightarrow { N}$ and ${p}: { E} \rightarrow { N}$
and a vertex-weight parameter $R$ and an edge-weight parameter $L$, is there a subset $S$ of $V$, such that $\sum_{e\in { E(S)}} p(e) \geq L$, and 
 $\sum_{v \in { S}} c(v) \le R$?
\end{definition}

\item The Partial vertex cover problem on bipartite graphs (PVCB) - This is the restriction of the partial vertex cover (PVC)
problem to bipartite graphs.

\item The Weighted Partial vertex cover problem on bipartite graphs (WPVCB) - This is the restriction of the weighted vertex cover (WVC) problem
to bipartite graphs.

\item The Partial vertex cover problem with matching constraint (PVCBM) - This is a variant of the PVCB problem, in which we are given
a third parameter $K_3$ and the goal is to find a vertex subset of size at most $K_1$, covering at least $K_2$ edges, such that the
edges covered include a matching of size at least $K_3$.
\end{enumerate}

The principal contributions of this paper are as follows:
\begin{enumerate}
\item Fixed-parameter tractability for a restricted version of the WPVCB problem.
\item ${\bf W}[1]$ hardness of the WPVCB problem with respect to a certain parameter.
\item Fixed parameter tractability of the Weighted partial vertex cover problem in bounded
degree graphs (not necessarily bipartite).
\item Fixed parameter tractability of the Weighted partial vertex cover problem with respect to $L$.
\item A parameterized algorithm for the matching variant of the PVCB problem.
\end{enumerate}

\section{Motivation and Related Work}

\label{motwork}
When $c\equiv 1$ and $p\equiv 1$, we get the well-known partial vertex
cover problem (PVC). PVC represents a natural theoretical
generalization of VC and is motivated by practical
applications. Flow-based risk-assessment models in computational
systems, for example, can be viewed as instances of PVC \cite{CGBS13}.

Although VC is polynomial-time solvable in bipartite graphs, the
Partial Vertex Cover problem on bipartite graphs is NP-hard. The
computational complexity of this problem has been open and recently
shown to be NP-hard \cite{Apoll, Joret, pvcbpaper, CS14}. Many
$2$-approximation algorithms for VC are known \cite{Vazirani}. There
is an approximation algorithm for the VC problem which has an
approximation factor of $2 - \theta(\frac{1}{\sqrt{logn}})$
\cite{Kar09}. This is the best known algorithm. The VC problem is also
known to be APX-complete \cite{PY91}. Moreover, it cannot be
approximated within a factor of $1.3606$ unless P = NP \cite{DS05},
and not within any constant factor smaller than $2$, unless the unique
games conjecture is false \cite{KR08}. Let us note that in
\cite{KPS11}, a $(\frac{4}{3}+\epsilon)$-approximation algorithm is
designed for WPVC for each $\epsilon >0$ when the input graph is
bipartite. This restriction is denoted by WPVCB.

All hardness results for the VC problem directly apply to the PVC
problem because the PVC problem is an extension of the VC
problem. Since 1990’s the PVC problem and the partial-cover variants
of similar graph problems have been extensively studied
\cite{BshB98,Bla03,KMR07,KLR08,KMRR06,M09}. In particular, there is an
$O(n \cdot \log n +m)$-time $2$-approximation algorithm based on the
primal-dual method \cite{M09}, as well as another combinatorial
2-approximation algorithm \cite{BFMR10}. Both of these algorithms are
for a more general soft-capacitated version of PVC. There are several
older 2-approximations resulting from different approaches
\cite{Bar01,BshB98,GKS04,Hoch98}. Let us also note that the WPVC
problem for trees (WPVCT) is studied in \cite{pvct}, where the authors
provide an FPTAS for it, and a polynomial time algorithm for the case
when vertices have no weights.

Another problem with close relationship to WPVC is the Budgeted
Maximum Coverage problem (BMC). In this problem one tries to find a
min-cost subset of vertices, such that the profit of covered edges is
maximized. In some sense, this problem can be viewed as a problem
``dual" to WPVC, and it can be shown that both problems are equivalent
from the perspective of exact solvability. The BMC problem for sets
(not necessarily graphs) admits a $(1 - \frac{1}{e})$-approximation
algorithm \cite{BMCPsets} but, special cases that beat this bound are
rare. The pipage rounding technique gives a
$\frac{3}{4}$-approximation algorithm for the BMC problem on graphs
\cite{AgeevSvirid} which is improved to $\frac{4}{5}$ for bipartite
graphs \cite{Apoll2}. Finally, let us note that in \cite{pvcbpaper,
  CS14}, an $8/9$-approximation algorithm for the problem is presented
when the input graph is bipartite and the vertices are unweighted. The
result is based on the linear-programming formulation of the problem,
and the constant $8/9$ matches the integrality gap of the linear
program.

In the present paper, we address these problems from the perspective
of fixed-parameter tractability (FPT). Recall that a combinatorial
problem $\Pi$ is said to be fixed-parameter tractable with respect to
a parameter $k$, if there is an algorithm for solving $\Pi$ exactly,
whose running time is bounded by $f(k)\cdot size^{O(1)}$. Here $f$ is
some (computable) function of $k$, and $size$ is the length of the
input. From the perspective of FPT, the PVC problem is in some sense
more difficult than the VC problem. For instance, the PVC problem is
W[1]-complete \cite{ParamBook15}, while the VC problem is FPT
\cite{Guo05,ParamBook15}. Related with this topic, let us note that in
\cite{AM11} the decision version of WPVCB is considered, where in a
bipartite graph one needs to check whether there is a subset of cost
at most $R$, whose coverage is at least $L$. The authors show that
this problem is FPT with respect to $R$, when the vertices and edges
of the bipartite graph are unweighted \cite{AM11}.

In this paper, by extending the methods of Amini et al., we show that
the decision version of WPVCB is FPT with respect to $R$, if the
vertices have cost one, however the edges may have weights. On the
negative side, the problem is $W[1]$-hard, even when edges have profit
one.  We complement this negative result by proving that for
bounded-degree graphs WPVC is FPT with respect to $R$. The same result
holds for the case of WPVCB when one is allowed to take only one
fractional vertex. We finish the paper by showing that WPVC is FPT
with respect to $L$. Terms and concepts that we do not define can be
found in \cite{ParamBook15}.

\section {Main Results}
\label{mainres}

In this section we present our results. Our goal is to investigate the
fixed-parameter tractability of the decision version of WPVCB
problem. It is formulated as follows:\\

{\bf WPVCBD:} Given a bipartite graph $B$, a cost function
$c:V(B)\rightarrow{N}$, a profit function $p: E(B)\rightarrow{N}$ and
positive integers $R$, $L$, the goal is to check whether there is a
set $S\subset V(B)$, such that $c(S)\leq R$ and $p(E(S))\geq L$.\\

When $c$ and $p$ are identically one, we get the PVCBD problem. When
$c$ is identically one, we get EPVCBD. Finally, when $p$ is
identically one, we get the VPVCBD problem. We will also use the same
scheme of notations when the input graph need not be bipartite. In
\cite{AM11}, PVCBD is considered and it is shown that the problem is
FPT with respect to $R$. Below we strengthen this result.

\begin{theorem}\label{thm:epvcbd} EPVCBD is FPT with respect to $R$.
\end{theorem}

\begin{proof} Roughly speaking we obtain the result with the approach of \cite{AM11} by considering the weighted degree instead of usual degree. Below we present the technical details.

Assume that we have an instance $I$ of EPVCBD. For a vertex $v$ of
$B$, let $\partial(v)$ be the set of edges of $B$ incident with
$v$. Define the set $S$ of vertices of $B$ as follows:
\[
S=\left\{v\in V(B): p(\partial(v))\geq \frac{L}{R}\right\}.
\]

We consider two cases based on the size of $S$.\\

Case 1: $|S|\geq 2R$. Consider the subgraph $H$ of $B$ induced by
$S$. Since $B$ is bipartite, $H$ is bipartite, too. Let $(X, Y)$ be
the bipartition of $H$, and assume that $|X|\geq |Y|$. Since
$|X|+|Y|=|S|\geq 2R$, we have $|X|\geq R$. Take any $R$ vertices of
$X$. Observe that $X$ is an independent set in $B$, hence the coverage
of the $R$ vertices is at least $L$. This means that $I$ is a ``yes"
instance.\\

Case 2: $|S|< 2R$. Observe that any feasible solution to $I$ must
intersect $S$. Hence, we do recursive guessing, that is, we try each
vertex of $S$ one by one as a possible vertex of the feasible
solution.\\

In the Case 1, the algorithm will run in polynomial time, so the most
expensive case is Case 2. Since the number of vertices in a feasible
solution is at most $R$, we have that the depth of the recursion is at
most $R$. Hence the total running time of our algorithm is
$O((2R)^{R}\cdot size^{O(1)})$.
\end{proof}

Our next result shows that WPVCBD and VPVCBD are W[1]-hard. Our
reduction is from the Multi-colored CLIQUE problem
\cite{ParamBook15}. It is formulated as follows:
\medskip

{\bf Multi-colored CLIQUE:} Given a graph $G$, a positive integer $k$
and a partition $(V_1,...,V_k)$ of vertices of $G$, the goal is to
check whether $G$ contains a $k$-clique $Q$, such that $Q$ contains
exactly one vertex from each $V_j$ for $j=1,...,k$. \\

Multi-colored CLIQUE is a well-studied problem which is known to be
W[1]-hard. Observe that since an edge $e$ connecting two vertices from
$V_i$ $1\leq i\leq k$ does not lie in a feasible clique, without loss
of generality, we can assume that for $i=1,...,k$ $V_i$ is an
independent set of vertices.

\begin{theorem}
WPVCBD is W[1]-hard parameterized by $R$.
\end{theorem}
\begin{proof}
We show an FPT-reduction from \textsc{Multi-colored Clique}. Let
$G=(V,E)$ be an instance of this problem with vertices partitioned as
$V=V_1 \cup \ldots \cup V_k$.  We create a bipartite graph $B=(U' \cup
V' \cup Z, E')$ as follows. Let $U'$ and $V'$ be two copies of $V$,
and let $V=V'=\{v_1,\ldots,v_n\}$, $U'=\{u_1,\ldots,u_n\}$, where for
each $i \in [n]$, $u_i$ is a copy of $v_i$. Here as usual
$[n]=\{1,\ldots,n\}.$ For a vertex $v \in V$ let $\chi(v)$ be its
color, i.e., $\chi(v)=i$ if $v \in V_i$, and extend this to $U' \cup
V'$ so that $\chi(u_i)=\chi(v_i)+k$.  For a vertex $x \in U' \cup V'$,
let the cost of $x$ be $c(x)=2^{\chi(x)}$.  Add an edge $u_iv_j$ to
$B$ if either $\chi(u_i)=\chi(v_j)+k$ and $i \neq j$, or $\chi(u_i)
\neq \chi(v_j)+k$ and $v_iv_j \notin E(G)$. Give all these edges
profit $1$.  Observe that a selection of one vertex from every color
class of $U' \cup V'$ forms an independent set in $B$, if and only if
it corresponds to two copies of a $k$-clique in $G$.

Add two additional vertices $z_1$ and $z_2$, let $Z=\{z_1, z_2 \}$ and
give both the cost $2^{2k+1}$.  Finally, for every vertex $x \in U'
\cup V'$ add an edge $xz_2$ or $z_1x$ (as appropriate, to maintain
bipartiteness) and give these edges a profit value so that the
\emph{total} profit of all edges incident with $x$ equals
$2^{\chi(x)}(n+1)+5^{\chi(x)}$.  (This is clearly possible, since the
total profit of all previously created edges incident with $x$ is
bounded by $n$.)

Set the budgets of the instance as vertex budget
\[
R = \sum_{i=1}^{2k} 2^i = 2^{2k+1}-2
\]
and profit threshold
\[
L=\sum_{i=1}^{2k} (2^i(n+1)+5^i)=(n+1)R + (5/4)(5^{2k}-1).
\]
This finishes the instance description. It is clear that the
construction can be performed in polynomial time, and the budget $R$
is a function of $k$.  It therefore remains to show that $(B,R,L)$ is
a positive instance of WPVCBD if and only if $G$ has a multi-colored
clique.

\emph{From the multi-colored clique problem to the partial vertex
  cover problem.}  Let $X \subseteq V(G)$ be a multi-colored
$k$-clique, and let $S=\{u_i, v_i \mid u_i \in X\}$.  Since $S$
contains one vertex for every color class of $B$ its total cost equals
$R$, and since it induces an independent set in $B$ the total profit
of the edges covered equals $L$.

\emph{From the partial vertex cover problem to the multi-colored
  clique problem.}  Now for the more challenging part of the
argument. We need to argue that the costs and profits balance out so
that the only way to select vertices to a total profit of $L$ is to
select one vertex from every color class of $B$.  For this, first
observe that for a vertex of class $i \in [2k]$, the ratio of the
total profit of its incident edges to its cost is
\[
\frac{2^i(n+1)+5^i}{2^i} = (n+1)+(5/2)^i.
\]
Let $S$ be a partial vertex cover of profit at least $L$ and cost at
most $R$, and for each $i \in [2k]$ let $n_i$ be the number of
vertices of $S$ with color class $i$.  We can consider the
contributions to the total profit in two parts.

From the first part of the above expression, it is clear that every
selection of cost at most $R$ contributes a profit of at most
$(n+1)R$, regardless of the distribution $n_i$.  Therefore we focus on
the contribution of the second part of the formula, with target profit
$L-(n+1)R$.

Considering this second part, define $R_t=\sum_{i=1}^t 2^i$ and
$L_t=\sum_{i=1}^t 5^i$. We show by induction that for every $t \in
[2k]$, the largest possible contribution of a selection of cost at
most $R_t$ is $L_t$ (which of course is achieved by making one
seletion per color class).  For $t=1$ this is trivial. Therefore, by
induction, let $t>1$ and assume that the claim holds for every value
$t' < t$.  Let $n'_i$, $i \in [t]$, denote the number of vertices
selected in color class $i$, for a selection of total cost at most
$R_t$.  Then if $n_t=0$, the maximum possible profit is
\[
R_t \cdot (5/2)^{t-1} < 2^{t+1} (5/2)^{t-1} = 4 \cdot 5^{t-1} < 5^t.
\]
Therefore, the total profit from the selection $n_i'$ is less than
that from a single vertex of class $t$.  Therefore $n_t' \geq 1$. But
then the remaining budget is $R_t-2^t=R_{t-1}$, and by induction the
optimal selection has $n_i'=1$ for every $i \in [t]$, completing the
induction step.  Therefore we may assume $n_i=1$ for every $i \in
[2k]$ for our selection $S$.

But then, finally, we observe that a total profit of $L$ is possible
only if $S$ is an independent set, since otherwise the profit of some
edge will have been double-counted in the above calculations.
Therefore, $S$ contains precisely one vertex of each color class of
$B$ forming an independent set. In particular, if $v_i \in S$ is a
selection in color class $j \in [k]$ for some $i \in [n]$, then the
selection in color class $j+k$ must be $u_i$. Also, for every pair of
color classes $i, i' \in [k]$ the selections in classes $i$ and $i'+k$
are independent in $B$ and therefore the selections in classes $i$ and
$i'$ are neighbors in $G$.  Thus $X=S \cap V'$ is a multi-colored
clique in $G$, as required.
\end{proof}

\begin{corollary}
The problem is W[1]-hard also in the variant where all edge profits
are 1, i.e., the VPVCBD problem is W[1]-hard.
\end{corollary}
\begin{proof}
The only edges of weight more than 1 in the above reduction are the
edges connecting to the special vertices $z_i$, and the largest edge
weight used is bounded by a function $f(k)(n+1)$. Therefore, instead
of using edge weights we can in FPT time simply create the
corresponding number of pendant vertices for each vertex. These
pendants can be given the same weight as the vertices $z_i$.
\end{proof}

A class of graphs is said to be bounded-degree, if there is a constant
$C$, such that all graphs from the class have maximum degrees at most
$C$. It turns out that when the input graphs have bounded-degree and
need not be bipartite, we have the following result:

\begin{theorem}\label{thm:FPTBoundedDegree} WPVCD is fixed-parameter tractable with respect to $R$ for bounded-degree graphs.
\end{theorem}

\begin{proof}
Let $I$ be an instance of the WPVCD problem, where $G=(V,E)$ is a
bounded-degree graph, $c: V \rightarrow N$, $p:E \rightarrow N$ are
cost and profit functions, and $L$, $R$ are constants. Assume that for
every $v \in V$, we have $d(v) \le d$. For $i=1,2,\ldots,R$, let
$M_i(G) = \{v : c(v) = i\}$ (we disregard the vertices of cost greater
than $R$). Choose a vertex $v_i\in M_i(G)$ which has the largest
coverage, and let $M$ be the set comprised of chosen vertices
$v_i$. The set of vertices of $G$ which have a neighbor from $M$, we
denote by $N$. Consider the set $M \cup N$. It is obvious that $|M
\cup N |\le (d+1) R$, where $d$ is the bound for the degree.

Let us show that if $I$ is a ``yes" instance, then we can construct a
feasible set, which intersects $M \cup N$. Indeed, let $S$ be a
feasible set, which does not have a vertex from $M \cup N$. Then we
can replace any vertex $v \in S$ with $v_{c(v)}$ to get a new set
$S'$. Since there is no vertex of $S$, which is a neighbor of
$v_{c(v)}$, we have that the profit of $S'$ does not
decrease. Moreover, as $c(v) = c(v_{c(v)})$, $S'$ is also
feasible. Observe that $S'$ intersects $M \cup N$.
 
Now we complete the proof by recursively guessing on $M \cup N$. Since
the number of vertices in a feasible set is at most $R$, we have that
the depth of the recursion is at most $R$, hence the total running
time of the algorithm is $((d+1)R)^{R} size^{O(1)}$.
\end{proof}

\begin{corollary} WPVCD is fixed-parameter tractable with respect to $R$ for graphs of maximum degree three, in particular, for cubic graphs.
\end{corollary}

Below we consider a version of the WPVCBD problem denoted by WPVCBFD
when one is allowed to take at most one vertex fractionally. When a
vertex $v$ is taken to an extent of $\alpha$, ($0 \leq \alpha \leq
1$), we assume that it contributes to the cost of the cover by
$\alpha\cdot c(v)$. The profits of covered edges are defined as
follows: if an edge $e$ is covered solely by $v$ then its profit is
$\alpha\cdot p(e)$, otherwise, it is $p(e)$. If a vertex $v$ is taken
fractionally, we say that $v$ is a fractional vertex.

\begin{theorem}\label{lemma:FractionalGroup} WPVCBFD is fixed-parameter tractable with respect to $R$.
\end{theorem}

\begin{proof} Let $I$ be an instance of the WPVCBFD problem, where $B=(V,E)$ is a bipartite graph, $c:V\rightarrow N$, $p:E \rightarrow N$ are cost and profit functions, and $L$, $R$ are constants. For any $v \in V$ let

$$S_B(v) = \{v_i : i=1,2,...,c(v)\}.$$

$S_B(v)$ is called the section corresponding to $v$. We create an
  instance $I'$ of the EPVCBD problem with rational edge profits as
  follows: the graph is $B'=(V',E')$, where
$$V' = \bigcup \limits_{v \in V} S_B(v),$$
$$E' = \{v_iu_j: \makebox{ } v_i \in S_B(v),\makebox{ } u_j \in
  S_B(u)\makebox{ } \text{and} \makebox{ } vu \in E\},$$ and the
  profit function is
$$p(v_iu_j) = \frac{p(vu)}{c(v)c(u)}.$$ The new problem can be turned
  to an instance of the EPVCBD problem with natural edge profits by
  multiplying all edge profits with the least common multiplier of all
  denominators of edge profits. Clearly, this does not change the
  hardness of the problem. By Theorem \ref{thm:epvcbd}, we know that
  EPVCBD is fixed-parameter tractable with respect to $R$. Hence $I'$
  can be solved in time $f(R)\cdot size^{O(1)}$.

Assume that we have a feasible set $S'$ in $I'$. We construct a
feasible set $S$ of $I$, so that it contains at most one fractional
vertex. Let $\alpha_v$ be defined as follows:

\[\alpha_v = \frac{|\{v_i : v_i\in S' \}|}{c(v)},\]
and take every vertex $v$ to an extent of $\alpha_v$. It is obvious
that we have a solution $S$ for $I$, however there can be many
vertices which are fractional. Recall that we should allow only one
such vertex.

Now, we are going to modify $S$. Let $v$ and $u$ be fractional
vertices, and assume that both $v_i \in S_B(v)$ and $u_j \in S_B(u)$
do not belong to $S'$. Without loss of generality, assume that $v$ has
higher coverage. Then, delete one vertex from $S'$ that lies in the
section of $u$ and add one vertex to $S'$ that belongs to the section
of $v$. We do this for every such $v$ and $u$. After this only one
vertex can be fractional.

It is obvious that we do only polynomially many steps after solving
$I'$. It means that our problem is fixed-parameter tractable with
respect to $R$.
\end{proof}


Exercise 5.11 from \cite{ParamBook15} implies that PVCD is
fixed-parameter tractable with respect to $L$. Below, we strengthen
the statement of this exercise by showing that the WPVCD problem can
be parameterized with respect to $L$.

\begin{theorem}\label{thm:FPTL} WPVCD is fixed-parameter tractable with respect to $L$.
\end{theorem}

\begin{proof}
Let $I$ be an instance of the WPVCD problem, where $G=(V,E)$ is a
graph, $c:V\rightarrow N$, $p:E \rightarrow N$ are cost and profit
functions, $L$ and $R$ are constants. We can assume that no vertex of
$G$ is isolated. Moreover, without loss of generality, we can assume
that for any vertex $v$, we have $c(v)\leq R$.

For every vertex $v$ we denote $p(v)=p(\partial(v)) = \sum_{e:v\in
  e}p(e)$ the total profit of edges incident to $v$. We can assume
that $p(v) \leq L-1$ for all vertices, as otherwise we will have a
feasible solution comprised of one vertex and, as a result, $I$ is a
``yes" instance. This, in particular, means that $d(v)\leq L-1$. For
$i=1,\ldots,L-1$ let $v_i$ be the vertex which has profit $i$ and for
any other vertex $u$, which has profit $i$, we have $c(v_i) \le
c(u)$. Let $M$ be the set of those vertices $v_i$. The set of vertices
of $G$ which have a neighbor from $M$ is denoted by $N$. It is obvious
that $|M \cup N |\le (L-1)+(L-1)^2<L^2$.

Let us show that if $I$ is a ``yes" instance and $S$ is a feasible set
in $I$, then we can construct a feasible set that intersects $M \cup
N$. Indeed, assume that $S$ does not contain any vertex from $M \cup
N$. Then we can replace any vertex $v \in S$ by the vertex
$v_{p(v)}$. Since there is no vertex in $S$ which is a neighbor of
$v_{p(v)}$, it follows that the total vertex cost has not increased,
and, as $p(v) = p(v_{p(v)})$, $S'$ is also feasible.

Now we complete the proof by recursively guessing on $M \cup N$. We
now show that the depth of the recursion is less than $2L$. For the
sake of contradiction, assume that during the recursive guessing, the
algorithm has considered the vertices $z_1,\ldots, z_{2L}$. Let
$Z=\{z_1,\ldots, z_{2L}\}$. Since the algorithm has considered these
vertices, we have that $c(Z)\leq R$. Then for the profit of edges
covered by $Z$, we will have the following bound:
\[
p(E(Z))\geq |E(Z)|\geq \frac{d(z_1)+\ldots+d(z_{2L})}{2}\geq
\frac{2L}{2}=L.
\]
Thus, $Z$ is a feasible set, hence $I$ is a ``yes" instance. This
means that there is no need to consider $2L$ or more vertices during
the recursive guessing. Hence the depth of the recursion is less than
$2L$. Since $|M \cup N |\le L^2$, we have that the running time of our
algorithm is bounded by $(L^2)^{2L} \cdot size^{O(1)}=L^{4L} \cdot
size^{O(1)}$.
\end{proof}

\section{The Matching problem}
\label{matching}
We now consider a variant of the PVCB problem. In this variant, we are
given a bipartite graph $G$ and three integers $k_1$, $k_2$ and $k_3$,
and the goal is to check whether there is a subset of $k_1$ vertices,
that covers at least $k_2$ edges, such that the covered edges contain
a matching of size at least $k_3$. This variant is called PVCBM.

Clearly, this problem is {\bf NP-hard}, as when $k_3=0$ it results into the
PVCB. Since PVCB is FPT with respect to $k_1$, it would be interesting
to parameterize this new version of the problem with respect to
$k_1$. Observe that we can assume that $k_3\leq k_1$ otherwise the
problem is a trivial NO-instance.

\begin{theorem}
 PVCBM is fixed-parameter tractable with respect to the parameters
 $k_1$, $k_2$ and $k_3$.
 \end{theorem}
 
 \begin{proof}
Let $PVCB(A,B)$ be the algorithm that in $FPT(A)$-time checks whether
there is a subset of $A$ vertices that covers at least $B$ edges of
the input bipartite graph. Now, assume that the graph $G$ and the
parameters $k_1$, $k_2$ and $k_3$ are given in the matching
problem. First, we run $PVCB(k_1,k_2)$. If there is no such subgraph,
then the answer to the matching problem is also negative. So we can
assume that $PVCB(k_1,k_2)$ returns such a subgraph. Next, by trying
$R=0,1,...,k_1$ we can find the smallest $R$ for which $PVCB(R,k_2)$
is a yes-instance.

Let $H$ be the edge-induced subgraph induced on these $\geq k_2$
edges. As usual, let $\nu(G)$ be the size of the largest matching in
$G$, and let $\tau(G)$ be the size of the smallest vertex cover in
$G$. By the classical K\"{o}nig theorem we have $\nu(G)=\tau(G)$ for
any bipartite graph $G$.

Observe that we can assume that $R<k_3\leq k_1$. To see this, observe
that $R$ represents the number of vertices required to cover all the
edges in $H$. In other words, it is a vertex cover of $H$. Thus, $R=
\nu(H)$ and since $H$ is a bipartite subgraph of $G$, $R$ is also the
size of a maximum matching in $H$. Thus, if $R \ge k_3$, then the
edges in $H$, which number at least $k_2$ can be covered by $R \le
k_1$ vertices and a matching of size $R \ge k_3$ is contained in $H$.

 Also, observe that if $\tau(G)<k_3$, then we have trivial
 NO-instance, as $G$ contains no matching of size $k_3$. Thus, we can
 assume that $\nu(G)=\tau(G)\geq k_3$.

Since $\tau(H)=R<k_3\leq \tau(G)$, we have that $E(H)\neq E(G)$. Thus,
there is as an edge $e$ lying outside $H$. Add $e$ to $H$. If
$\tau(H)$ has increased by adding $e$, define $R:=R+1$, otherwise let
$R$ be the same. Repeat this process of adding edges outside
$H$. Since $\tau(H)=R<k_3\leq \tau(G)$, at some point we will arrive
into $H$ such that $R=\tau(H)=k_3\leq k_1$. Observe that $H$ can be
covered with at most $k_1$ vertices, it has at least $k_2$ edges and
it contains a matching of size $k_3$. Thus, the problem is a
YES-instance.

Finally, let us observe that the running-time of this algorithm is FPT
in $k_1$. We need at most $k_1$ calls of $PVCB(k_1, k_2)$. Since the
latter is FPT with respect to $k_1$, we have ther result.
\end{proof}

\section{Conclusion}
\label{conc}
 In this paper, we studied the partial vertex cover problem from the perspectives of
 parameterized tractability and {\bf W[1]-hardness}. Although our primary focus
 was on bipartite graphs, we obtained new results for the general case as well.
 Our main contributions include showing that a restricted version of the WPVCB problem
 is fixed-parameter tractable and that this problem is {\bf W[1]-hard}, with respect
 to the edge-weight parameter. We also showed that the WPVC problem is fixed-parameter
 tractable in bounded graphs. Finally, we introduced a new variant of the partial vertex
 cover problem called PVCBM and showed that it is fixed-parameter tractable.

\section*{Acknowledgements}
We are grateful to Magnus Wahlstr\"{o}m for his proof of Theorem $3$.



\bibliographystyle{elsarticle-num}



\end{document}